\documentclass[prl,nopacs,twocolumn,10pt]{revtex4}
\usepackage{amsmath}
\usepackage{graphicx,epsfig,amsmath,latexsym,amssymb,verbatim,color}
\usepackage{dsfont}
\usepackage{theorem}
\usepackage{url}
 \usepackage{hyperref}
\hypersetup{colorlinks=true,citecolor=blue,linkcolor=blue,filecolor=blue,urlcolor=blue,breaklinks=true}

\newtheorem{definition}{Definition}
\newtheorem{proposition}[definition]{Proposition}

\def\squareforqed{\hbox{\rlap{$\sqcap$}$\sqcup$}}
\def\qed{\ifmmode\squareforqed\else{\unskip\nobreak\hfil
\penalty50\hskip1em\null\nobreak\hfil\squareforqed
\parfillskip=0pt\finalhyphendemerits=0\endgraf}\fi}
\def\endenv{\ifmmode\;\else{\unskip\nobreak\hfil
\penalty50\hskip1em\null\nobreak\hfil\;
\parfillskip=0pt\finalhyphendemerits=0\endgraf}\fi}
\newenvironment{proof}{\noindent \textbf{{Proof~} }}{\qed}

\mathchardef\ordinarycolon\mathcode`\:
\mathcode`\:=\string"8000
\def\vcentcolon{\mathrel{\mathop\ordinarycolon}}
\begingroup \catcode`\:=\active
  \lowercase{\endgroup
  \let :\vcentcolon
  }

\newcommand{\nc}{\newcommand}
\nc{\rnc}{\renewcommand}
\nc{\beg}{\begin{equation}}
\nc{\eeq}{{\end{equation}}}
\nc{\beqa}{\begin{eqnarray}}
\nc{\eeqa}{\end{eqnarray}}
\nc{\lbar}[1]{\overline{#1}}
\nc{\bra}[1]{\langle#1|}
\nc{\ket}[1]{|#1\rangle}
\nc{\ketbra}[2]{|#1\rangle\!\langle#2|}
\nc{\braket}[2]{\langle#1|#2\rangle}

\nc{\proj}[1]{| #1\rangle\!\langle #1 |}
\nc{\avg}[1]{\langle#1\rangle}
\nc{\Rank}{\operatorname{Rank}}
\nc{\smfrac}[2]{\mbox{$\frac{#1}{#2}$}}
\nc{\tr}{\operatorname{Tr}}
\nc{\ox}{\otimes}
\nc{\dg}{\dagger}
\nc{\dn}{\downarrow}
\nc{\cA}{{\cal A}}
\nc{\cB}{{\cal B}}
\nc{\cC}{{\cal C}}
\nc{\cD}{{\cal D}}
\nc{\cE}{{\cal E}}
\nc{\cF}{{\cal F}}
\nc{\cG}{{\cal G}}
\nc{\cH}{{\cal H}}
\nc{\cI}{{\cal I}}
\nc{\cJ}{{\cal J}}
\nc{\cK}{{\cal K}}
\nc{\cL}{{\cal L}}
\nc{\cM}{{\cal M}}
\nc{\cN}{{\cal N}}
\nc{\cO}{{\cal O}}
\nc{\cP}{{\cal P}}
\nc{\cQ}{{\cal Q}}
\nc{\cR}{{\cal R}}
\nc{\cS}{{\cal S}}
\nc{\cT}{{\cal T}}
\nc{\cX}{{\cal X}}
\nc{\cY}{{\cal Y}}
\nc{\cZ}{{\cal Z}}
\nc{\cW}{{\cal W}}
\nc{\csupp}{{\operatorname{csupp}}}
\nc{\qsupp}{{\operatorname{qsupp}}}
\nc{\var}{{\operatorname{var}}}
\nc{\rar}{\rightarrow}
\nc{\lrar}{\longrightarrow}
\nc{\polylog}{{\operatorname{polylog}}}
\nc{\1}{{\mathds{1}}}
\nc{\wt}{{\operatorname{wt}}}
\nc{\av}[1]{{\left\langle {#1} \right\rangle}}
\nc{\supp}{{\operatorname{supp}}}

\def\a{\alpha}

\def\G{\Gamma}

\def\O{\Omega}

\nc{\RR}{{{\mathbb R}}}
\nc{\CC}{{{\mathbb C}}}
\nc{\FF}{{{\mathbb F}}}
\nc{\NN}{{{\mathbb N}}}
\nc{\ZZ}{{{\mathbb Z}}}
\nc{\PP}{{{\mathbb P}}}
\nc{\QQ}{{{\mathbb Q}}}
\nc{\UU}{{{\mathbb U}}}
\nc{\EE}{{{\mathbb E}}}
\nc{\id}{{\operatorname{id}}}

\nc{\CHSH}{{\operatorname{CHSH}}}

\nc{\be}{\begin{equation}}
\nc{\ee}{{\end{equation}}}
\nc{\bea}{\begin{eqnarray}}
\nc{\eea}{\end{eqnarray}}
\nc{\<}{\langle}
\rnc{\>}{\rangle}
\nc{\Hom}[2]{\mbox{Hom}(\CC^{#1},\CC^{#2})}
\nc{\rU}{\mbox{U}}

\nc{\ob}[1]{#1}

\nc{\SEP}{{\text{SEP}}}
\nc{\NS}{{\text{NS}}}
\nc{\LOCC}{{\text{LOCC}}}
\nc{\PPT}{{\text{PPT}}}
\nc{\EXT}{{\text{EXT}}}
\nc{\Sym}{{\operatorname{Sym}}}
\nc{\ERLO}{{E_{\text{r,LO}}}}
\nc{\ERLOCC}{{E_{\text{r,LOCC}}}}
\nc{\ERPPT}{{E_{\text{r,PPT}}}}
\nc{\ERLOCCinfty}{{E^{\infty}_{\text{r,LOCC}}}}
\nc{\Aram}{{\operatorname{\sf A}}}

\begin{document}
\title{Irreversibility of Asymptotic Entanglement Manipulation Under Quantum Operations Completely Preserving Positivity of Partial
Transpose}
\author{Xin Wang$^{1}$}
\email{xin.wang-8@student.uts.edu.au}
\author{Runyao Duan$^{1,2}$}
\email{runyao.duan@uts.edu.au}

\affiliation{$^1$Centre for Quantum Software and Information, Faculty of Engineering and Information Technology, University of Technology Sydney, NSW 2007, Australia}
\affiliation{$^2$UTS-AMSS Joint Research Laboratory for Quantum Computation and Quantum Information Processing, Academy of Mathematics and Systems Science, Chinese Academy of Sciences, Beijing 100190, China}

\begin{abstract}
We demonstrate the irreversibility of asymptotic entanglement manipulation under quantum operations that completely preserve the positivity of partial transpose (PPT), resolving a major open problem in quantum information theory. Our key tool is a new efficiently computable additive lower bound for the asymptotic relative entropy of entanglement with respect to PPT states, which can be used to evaluate the entanglement cost under local operations and classical communication (LOCC). We find that for any rank-two mixed state supporting on the $3\otimes3$ antisymmetric subspace, the amount of distillable entanglement by PPT operations is strictly smaller than one entanglement bit (ebit) while its entanglement cost under PPT operations is exactly one ebit. As byproduct, we find that for this class of states, both the Rains' bound and its regularization, are strictly less than the asymptotic relative entropy of entanglement. So, in general, there is no unique entanglement measure for the manipulation of entanglement by PPT operations. We further show a computable sufficient condition for the irreversibility of entanglement distillation by LOCC (or PPT) operations. 
\end{abstract}
\maketitle
\textbf{Introduction:} 
In quantum information science,
the resource theory of entanglement studies the interconvertibilities of entanglement under restricted classes of allowed operations. The irreversibility is crucial to this entanglement resource theory and it was sometimes argued to be the difference between entanglement and thermodynamics, as the Carnot cycle is reversible. When local operations and classical communication (LOCC) is available,
the manipulation of entanglement is irreversible in the finite-copy regime. More precisely, the amount of pure entanglement that can be distilled from a finite number of copies of a state $\rho$ is usually strictly smaller than the amount of pure entanglement needed to prepare the same number of copies of $\rho$~\cite{Bennett1996c}. Surprisingly, in the asymptotic limit of an arbitrarily large number of copies
of the bipartite pure states, this process is shown to be reversible~\cite{Bennett1996b}. 
But for mixed states, this asymptotic reversibility does not hold anymore~\cite{Vidal2001,Vidal2001a,Vidal2002b,Vollbrecht2004,Cornelio2011}. In particular, one requires a positive rate of pure states to generate the bound entanglement by LOCC~\cite{Vidal2001,Yang2005}, while it is well known that no pure state can be distilled from it~\cite{Horodecki1998}.

Various approaches have been considered to enlarge the class of operations to ensure reversible interconversion of entanglement in the asymptotic regime. A natural candidate is the class of quantum operations that completely preserve positivity of partial transpose (PPT)~\cite{Rains2001}, which include all quantum operations that can be implemented by LOCC. A remarkable result is that any state with a nonpositive partial transpose (NPT) is distillable under this class of operations~\cite{Eggeling2001}. This suggests the possibility of reversibility under PPT operations, and there are examples of mixed states which can be reversibly converted into pure states in the asymptotic setting, e.g. the class of antisymmetric states of arbitrary dimension~\cite{Audenaert2003}.  It is noteworthy that for tripartite states, Ishizaka and Plenio~\cite{Ishizaka2005}  showed that asymptotic entanglement manipulation is irreversible. However, for bipartite states, the reversibility under PPT operations remained unknown so far since there were no further examples.  Recently, a reversible theory of entanglement considering all asymptotically non-entangling transformations was studied in Refs.~\cite{Brandao2008,Brandao2010} and the unique entanglement measure is identified as the asymptotic  relative entropy of entanglement. A more general reversible framework for quantum resource theories was recently introduced in Ref.~\cite{Brandao2015}.

When the unit of pure entanglement is set to be the standard $2\otimes 2$ Bell pair $1/\sqrt{2}(\ket{00}+\ket{11})$, or entanglement bit (ebit),  two fundamental ways of entanglement manipulation are well known, namely, entanglement distillation and entanglement dilution~\cite{Bennett1996b,Bennett1996c}. These two tasks also naturally raise two fundamental entanglement measures: distillable entanglement and entanglement cost~\cite{Bennett1996c}. To be specific, distillable entanglement is the highest rate at which one can obtain Bell pairs from the given state under allowed operations, while entanglement cost is the lowest rate for converting Bell pairs to the given state. It is worth noting that if one can show a gap between the distillable entanglement and entanglement cost under PPT operations, then it will lead to the  irreversibility of asymptotic entanglement manipulation. However, this problem is still very hard since for general mixed states it is highly nontrivial to evaluate these two measures both of which are given by a limiting procedure. 

In this Letter, we demonstrate that irreversibility still exists in the asymptotic entanglement manipulation under PPT operations, which resolves a long-standing open problem in quantum information theory \cite{Audenaert2003,Horodecki2002,M.Plenio}. Our approach is to show a gap between the regularized Rains' bound and the asymptotic relative entropy of entanglement \cite{Audenaert2001} with respect to PPT states, which also resolves another open problem in Ref. \cite{Plenio2007}. More precisely, we introduce an additive semidefinite programming (SDP)  lower bound for the asymptotic relative entropy of entanglement with respect to PPT states. With this lower bound, we are able to show that the PPT-entanglement cost of any rank-two state supporting on the $3\ox3$ antisymmetric subspace is exactly one ebit while its PPT-distillable entanglement is strictly smaller than one. As a corollary, we show that there is no unique entanglement measure under PPT operations. This means that entanglement theory under PPT operations differs from thermodynamics, since in the second law of thermodynamics, the entropy uniquely determines whether a state is adiabatically accessible from another. We also give a sufficient condition to efficiently verify the irreversibility of entanglement distillation by LOCC (or PPT) operations. A general class of states are constructed to illustrate this phenomenon, see FIG.~\ref{irr_ppt}.
\begin{figure}[ht]
\centering
\vspace{-0.35cm}
\includegraphics[width=0.40\textwidth]{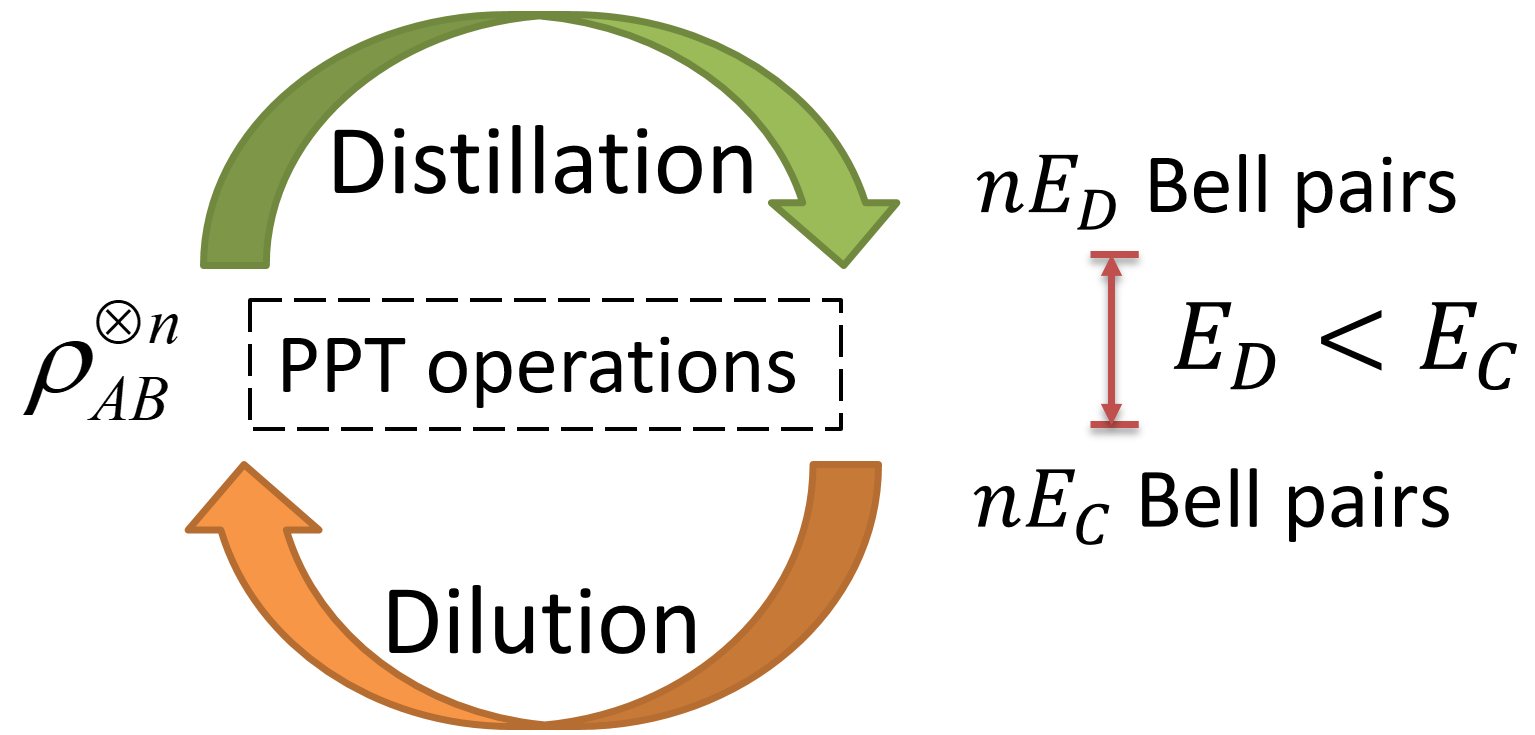}
\caption{The amount of Bell pairs distilled from the state is insufficient to reconstruct the given state under PPT operations in the asymptotic regime.}\label{irr_ppt}
\vspace{-0.35cm}
\end{figure}
%Asymptotic entanglement manipulation under PPT operations is irreversible.

Before we present our main results, let us review some notations and preliminaries. 
We will use symbols such as $A$ (or $A'$) and $B$ (or $B'$) to denote (finite-dimensional) Hilbert spaces associated with Alice and Bob, respectively. The set of linear operators over $A$ is denoted by $\cL(A)$. For a linear operator $R$ over a Hilbert space, we define $|R|=\sqrt{R^\dagger R}$ and the trace norm $\|R\|_1=\tr |R|$, where $R^\dagger$ is the Hermitian conjugate of $R$. The operator norm $\|R\|_\infty$ is defined as the maximum eigenvalue of $|R|$.  A deterministic quantum operation $\cN$ from $A'$ to $B$ is simply a completely positive and trace-preserving (CPTP) linear map from $\cL(A')$ to $\cL(B)$.  A positive semidefinite operator $E_{AB} \in \cL(A\ox B)$ is said to be PPT if $E_{AB}^{T_{B}}\geq 0$, where $T_B$ means the partial transpose over the system $B$, i.e., 
$(\ketbra{i_Aj_B}{k_Al_B})^{T_{B}}=\ketbra{i_Al_B}{k_Aj_B}$. 
%The Choi-Jamio\l{}kowski matrix of $\cN$ is given by $J_{AB}=\sum_{ij} \ketbra{i_A}{j_{A}} \ox \cN(\ketbra{i_{A'}}{j_{A'}})$, where $\{\ket{i_A}\}$ and $\{\ket{i_{A'}}\}$ are orthonormal basis on isomorphic Hilbert spaces $A$ and $A'$, respectively.  A bipartite operation is said to be a PPT operation if its  Choi-Jamio\l{}kowski matrix is PPT.
%The class of PPT operations \cite{Rains2001} include all quantum operations
%that can be implemented by LOCC.

The task of entanglement distillation aims at obtaining maximally entangled states such as Bell pairs from less-entangled bipartite states. Imagine that Alice and Bob share a large supply of identically prepared state, and they want to convert these states to high fidelity Bell pairs using $\O$ operation. (We use $\Omega$ to represent one of \text{LOCC or PPT} through out the paper.) The \emph{distillable entanglement} $E_{D,\O}$ of $\rho$ quantifies the optimal rate $r$ of converting $\rho^{\ox n}$ to  ${rn}$  Bell pairs with an arbitrarily high fidelity in the limit of large $n$. The reverse task is entanglement dilution. At this time, Alice and Bob share a large supply of Bell pairs and they are  to convert $rn$ Bell pairs  to $n$ high fidelity copies of the desired state $\rho^{\ox n}$. The \emph{entanglement cost} $E_{C,\O}$ quantifies the optimal rate $r$ of converting ${rn}$  Bell pairs to $\rho^{\ox n}$  with an arbitrarily high fidelity in the limit of large $n$. 

For simplicity, we denote $E_{D,PPT}$ and $E_{C,PPT}$ as $E_{D}$ and $E_{C}$, respectively.  For entanglement cost, Hayden, Horodecki and Terhal~\cite{Hayden2001} proved that $E_{C,LOCC}$ equals to the regularized entanglement of formation~\cite{Bennett1996c} while the similar result is not true for PPT operations. For distillable entanglement, the best known bound is
the Rains' bound \cite{Rains2001} and it is reformulated in Ref.~\cite{Audenaert2002} as the following convex optimization problem: 
\begin{equation}
R(\rho)=\min S(\rho||\sigma) \text{ s.t. } \ \sigma\ge0, \|\sigma^{T_B}\|_1\le 1.
\end{equation}
In  this formula, $S(\rho||\sigma)=\tr (\rho\log \rho-\rho\log\sigma)$ denotes
the quantum relative entropy, where we take $\log \equiv \log_2$ throughout the paper. The regularized Rains' bound, i.e., $R^{\infty}(\rho)=\inf_{n\geq 1}  R(\rho^{\ox n})/n$, was first introduced in Ref.~\cite{Hayashi2006a}. Very recently we showed that the Rains' bound is not additive even for a class of two-qubit states \cite{Wang2016c}. The regularized Rains' bound is thus a better upper bound for the distillable entanglement.  

The PPT-relative entropy of entanglement (REE)~\cite{Vedral1997, Vedral1998a, Vedral1997a}  is defined by
\begin{align*}
E_{R}(\rho)= \min S(\rho || \sigma) \ \text{ s.t. }\  \sigma, \sigma^{T_B}\ge 0,\tr \sigma=1.
\end{align*}
And the asymptotic PPT-relative entropy of entanglement is given by $E_{R}^{\infty}(\rho)=\inf_{n\geq 1} E_{R}(\rho^{\ox n})/n.$
It was shown in Ref. \cite{Hayashi2006a} that the asymptotic REE is indeed a lower bound to the PPT-entanglement cost.  Then, for a general quantum state $\rho$, it always holds that 
\begin{equation}\label{ED REE EC}
E_D(\rho)\leq R^{\infty}(\rho)\leq E^{\infty}_R(\rho)\leq E_{C}(\rho).
\end{equation}
And it has been open for years whether any of these inequalities could be strict.

The main contribution of this Letter is to show that the second inequality is strict for a class of rank-two states supporting on the $3\otimes 3$ antisymmetric subspace. As the first example, let us consider $\rho_{v}=\frac{1}{2}(\proj{v_1}+\proj{v_2})$ with 
$$
\ket {v_1}={1}/{\sqrt 2}(\ket {01}-\ket{10}), \ket {v_2}={1}/{\sqrt 2}(\ket {02}-\ket{20}).
$$
The projection onto $\supp(\rho_v)$ is $P_v=\proj{v_1}+\proj{v_2}$. In Ref. \cite{Chitambar2009}, Chitambar and one of us showed that this state can be transformed into some $2\otimes 2$ pure entangled state by a suitable separable operation while no finite-round LOCC protocol can do that. Here we show that
%The projection onto the support of $\rho_v$ is given by 
%$P_{AB}=\proj{v_1}+\proj{v_2}$. 
\begin{equation}\label{mainresult}
E_{D}(\rho_v)=R^{\infty}(\rho_v)<E_{R}^{\infty}(\rho_v)=E_{C}(\rho_v). 
\end{equation}
That means the asymptotic entanglement manipulation of $\rho_v$ under PPT operations is irreversible, thus resolving a long-standing open problem in quantum information theory~\cite{Audenaert2003,Horodecki2002,M.Plenio}. Furthermore, it also answers another open problem in Ref.~\cite{Plenio2007} by showing a nonzero gap between the regualized Rains' bound and the asymptotic REE of $\rho_v$. The proofs are clear from Propositions \ref{PPT cost} and \ref{PPT distill}  below.

\textbf{An SDP lower bound for  $E_{R}^{\infty}(\rho)$:}
Our key tool is an efficiently computable additive lower bound for the asymptotic REE. In the one-copy case,  we need to do some relaxations of the minimization of $S(\rho||\sigma)$ with respect to PPT states.  Note that the support of a state $\rho$, denoted by $\supp(\rho)$, is a subspace spanned by the eigenvectors of $\rho$ with positive eigenvalues.  Let $D(\rho)=\{\rho': \supp(\rho')\subseteq\supp(\rho)\}$ be the set of quantum states whose supports are contained in that of $\rho$, and let $\G$ be the set of PPT states.  We can first relax the minimization of  $S(\rho||\sigma)$  to the smallest relative entropy distance between $D(\rho)$ and  the set $\G$. See FIG. 2 below for an intuitive illustration of the ideas. 
\begin{figure}[htbp]
\includegraphics[width=0.38\textwidth]{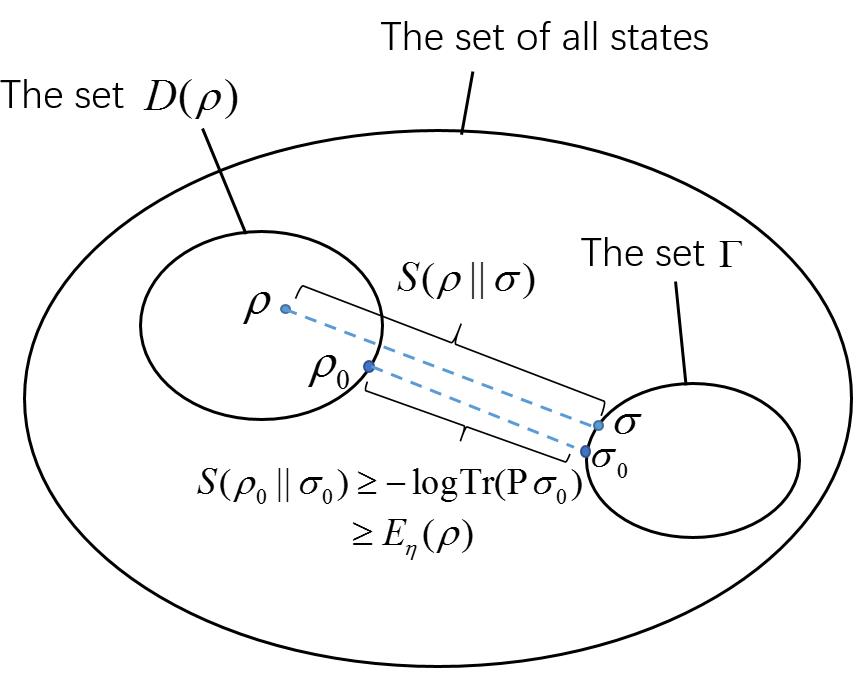}
 \caption{The PPT-relative entropy of entanglement is defined as the smallest quantum relative entropy  from  $\rho$ to the state $\sigma$ taken from the set of PPT states $\G$. Assume that $\rho_0$ and $\sigma_0$ give the smallest quantum relative entropy from  $D(\rho)$ to $\G$. It is clear that $E_{R}(\rho)=S(\rho||\sigma)\ge S(\rho_0||\sigma_0)$ and we show $S(\rho_0||\sigma_0)\ge -\log \tr P\sigma_0 \ge E_\eta(\rho)$ in Proposition \ref{E eta bound ER}, where $P$ is the projection onto the support of $\rho$.  This lower bound $E_\eta(\rho)$ is powerful since it still works in the asymptotic setting due to its additivity under tensor product.
 }
\end{figure}

Then applying properties of quantum relative entropy, we can further relax the problem to minimizing $-\log \tr P\sigma$ over all PPT states $\sigma$, where $P$ is the projection onto $\supp(\rho)$. Noting that this is SDP-computable, we can use SDP techniques to obtain the following bound
\begin{equation}\begin{split}\label{E_eta}
E_\eta(\rho)=\max  -\log\|Y^{T_B} \|_{\infty}, \text{ s.t. } -Y\le P^{T_B}\le Y.
\end{split}\end{equation}
Interestingly, $E_\eta(\cdot)$ is additive under tensor product, i.e.,  
$$E_\eta(\rho_1\ox\rho_2)=E_\eta(\rho_1)+E_\eta(\rho_2),$$
so we can overcome the difficulty of estimating  the regularised relative entropy of entanglement.  
The additivity of  $E_\eta(\cdot)$ can be proved by utilizing the duality theory of SDP \cite{Vandenberghe1996,Watrous2011b}. The detailed proof can be found in the supplementary material. This $E_\eta$ can be efficiently computed since SDP can be solved by efficient algorithms~\cite{Alizadeh1995} and it can also be implemented via CVX~\cite{Grant2008}  and QETLAB~\cite{NathanielJohnston2016}. In particular, $E_\eta$ becomes a computable lower bound for the entanglement cost under LOCC (or PPT) operations.

\begin{proposition}\label{E eta bound ER}
For any state $\rho$,
\begin{equation}
 E_{R}^{\infty}(\rho) \ge E_\eta(\rho).
\end{equation}
Consequently, $E_{C,LOCC}(\rho)\ge E_C(\rho) \ge E_\eta(\rho)$.
\end{proposition}
\begin{proof}
Firstly, let us introduce a CPTP map by $\cN(\tau)=P\tau P+(\1-P)\tau(\1-P)$. Then  
for $\rho_0\in D(\rho)$ and $\sigma_0\in \G$, we have that
\begin{equation}\begin{split}
S(\rho_0||\sigma_0)&\ge S(\cN(\rho_0)||\cN(\sigma_0))\\
                &=S(\rho_0||{P\sigma_0 P}/{\tr P\sigma_0 P})-\log\tr P\sigma_0\\
                &\ge -\log\tr P\sigma_0,
\end{split}\end{equation}
where the first inequality is from the monotonicity of quantum relative entropy  \cite{Lindblad1975,Uhlmann1977}  and the second inequality
 is due to the non-negativity of quantum relative entropy.
Therefore,
$$\min_{\sigma\in \G} S(\rho||\sigma)\ge \min_{\rho_0\in D(\rho), \sigma_0\in \G} S(\rho_0||\sigma_0) \ge \min_{\sigma_0\in \G} -\log\tr P\sigma_0.$$
This step transforms the problem to SDP problems and it can also be proved via the min-relative entropy in \cite{Datta2009}.

Secondly, utilizing the weak duality of SDP~\cite{Vandenberghe1996} (see the supplementary material for details), we have that 
\begin{align*}
\max_{\sigma_0\in \G} \tr P\sigma_0 &\le \min t \text{ s.t. } Y^{T_B}\le t\1, P^{T_B}\le Y\\
&\le \min t \text{ s.t.} -t\1\le Y^{T_B}\le t\1,  -Y\le P^{T_B}\le Y\\
& = \min \|Y^{T_B} \|_{\infty}, \text{ s.t. } -Y\le P^{T_B}\le Y.
\end{align*}
Thus, 
$$E_{R}(\rho) \ge -\log\max_{\sigma_0\in \G} \tr P\sigma_0 \ge E_\eta(\rho).$$

Finally, noting that $E_\eta(\rho)$ is additive, we have that 
\begin{align*}
 E_{R}^{\infty}(\rho)&=\inf_{n\geq 1}E_{R}(\rho^{\ox n})/n\\
 &\ge\inf_{n\geq 1}  E_\eta(\rho^{\ox n})/n =E_\eta(\rho).
 \end{align*}
 By Eq. (\ref{ED REE EC}), we have  $E_C(\rho) \ge E_\eta(\rho)$.
\end{proof}

\textbf{PPT-entanglement cost of $\rho_v$:}
Applying the lower bound $E_\eta(\rho)$, we are now ready to show that the PPT-entanglement cost of $\rho_v$ is still one ebit.
\begin{proposition}\label{PPT cost}
For state $\rho_v$,
$E_{C}(\rho_v)=E_{R}^{\infty}(\rho_v)=1$.
\end{proposition}
\begin{proof}
Firstly, suppose that $Q=\proj{01}+\proj{10}+\proj{02}+\proj{20}$ and we can prove that $E_\eta(\rho_v)\le E_{R}^{\infty}(\rho_v)\le {1}$ by choosing a PPT state
$\tau=Q/4$ such that $S(\rho_v || \tau)=1$.  

Secondly, we are going to prove $E_\eta(\rho_v)\ge 1$.
To see this, suppose that
\begin{align*}
Y=&\frac{1}{2}(Q+\proj{00}
+(\ket {11}+\ket{22})(\bra {11}+\bra{22})).
\end{align*}
Noting that $$Y-P_v^{T_B}=\frac{1}{2}(\ket {00}+\ket{11}+\ket{22})(\bra {00}+\bra{11}+\bra{22}),$$ it is clear that 
$P^{T_B}\le Y$. Moreover,
\begin{align*}
Y+P_v^{T_B}=Q+\frac{1}{2}(\ket {00}-\ket{11}-\ket{22})(\bra {00}-\bra{11}-\bra{22}),
\end{align*}
which means that 
$P_v^{T_B}\ge -Y$.

Then $Y$ is a feasible solution to the SDP (\ref{E_eta}) of $E_\eta(\rho_v)$.
Thus, 
$$E_\eta(\rho_v)\ge -\log\|Y^{T_B}\|_{\infty}=-\log 1/2=1,$$ and
 we can conclude that $E_\eta(\rho_v)=E_{R}^{\infty}(\rho_v)=1$.

Finally, it is obvious that a standard Bell pair is sufficiently to prepare an exact copy of $\rho_v$ by LOCC. Combining with the above bounds, we have that $1=E_{\eta}(\rho_v)\leq E_{R}^{\infty}(\rho_v)\leq E_{C}(\rho_v)\leq E_{C,LOCC}(\rho_v)\leq 1$.
\end{proof}

It is worth pointing out that our approach to evaluating the PPT-entanglement cost is to combine the lower bound $E_{\eta}$ and the upper bound $E_{C, LOCC}$.
This result provides a new proof of the entanglement cost of the rank-two $3\ox3$ antisymmetric state in Ref.~\cite{Yura2003}. Moreover, our result is stronger as it shows that the entanglement cost under PPT operations of this state is still one ebit.

\textbf{PPT-distillable entanglement of $\rho_v$:}
We can evaluate the PPT-distillable entanglement of $\rho_v$ by the upper bound of Rains' bound and the SDP characterization of the one-copy deterministic PPT-distillable entanglement~\cite{Wang2016}.
\begin{proposition}\label{PPT distill}
\begin{equation}
E_{D}(\rho_v)=R^{\infty}(\rho_v)=\log (1+{1}/{\sqrt 2}).
\end{equation}
\end{proposition}

\begin{proof}
Firstly, we need to introduce upper and lower SDP bounds to evaluate the distillable entanglement and the regularized Rains' bound.  The logarithmic negativity~\cite{Vidal2002,Plenio2005b}  is an upper bound on PPT-distillable entanglement, i.e., 
$E_N(\rho)=\log \|\rho^{T_B}\|_1$.

The following one-copy deterministic PPT-distillable entanglement was also obtained in Ref.~\cite{Wang2016,Fang2017}, 
\begin{equation}\label{prime PPT one-copy}
\begin{split}
E_{0,D}^{(1)}(\rho)= \max& -\log_2 \|R^{T_B}\|_\infty, \\
\text{ s.t. }  &  P\le R\le \1.
\end{split}
\end{equation}
where $P$ is the projection onto the support of $\rho$.  Clearly $E_{0,D}^{(1)}(\rho)$ is efficiently computable by SDP, and for a general bipartite state $\rho$ we have 
$$E_{0,D}^{(1)}(\rho)\le E_{D}(\rho)\leq R^{\infty}(\rho)\leq E_N(\rho),$$
which is very helpful to determine the exact values of PPT-distillable entanglement  for some states.

Now it is easy to check that $\|\rho_v^{T_B}\|_1=1+{1}/{\sqrt 2}$. Then,
\begin{equation}\label{ED upper}
R^{\infty}(\rho_v) \le E_N(\rho_v) \le  \log (1+{1}/{\sqrt 2}).
\end{equation}

On the other hand, let
$$R=(3-2\sqrt 2)(\proj{r_1}+\proj{r_2})+P_v$$
with 
$ \ket{r_1}=(\ket{01}+\ket{10})/\sqrt{2} $ and $ \ket{r_2}=(\ket{02}+\ket{20})/\sqrt{2}$.
It is easy to check that $P_v\le R\le \1$, which means that $R$ is a feasible solution to SDP (\ref{prime PPT one-copy}) of  $E_{0,D}^{(1)}(\rho_v)$.  
Therefore,
\begin{equation}\label{ED lower}
E_{0,D}^{(1)}(\rho_v) \ge -\log \| R^{T_B}\|_{\infty} =\log(1+{1}/{\sqrt 2}).
\end{equation}

Finally,  combining Eq. (\ref{ED upper}) and Eq. (\ref{ED lower}), we have that $E_{D}(\rho_v)=R^{\infty}(\rho_v)=\log (1+{1}/{\sqrt 2})$.
\end{proof}

\textbf{General irreversibility under PPT operations:} We have shown the irreversibility of the  entanglement distillation of $\rho_v$ under PPT operations. One can use similar technique to prove this irreversibility for any $\rho$ with spectral decomposition 
\begin{equation}\rho=p\proj{u_1}+(1-p)\proj{u_2} \ (0<p<1),\end{equation}
 where $\ket {u_1}=(\ket{01}-\ket{10})/\sqrt2, \ket {u_2}=(\ket{ab}-\ket{ba})/\sqrt2$ and $\braket{u_1}{u_2}=0$. Interestingly,  it holds that $ E_{D}(\rho)<1=E_{C}(\rho)$. (See the supplemenrary material). More generally, we can provide a sufficient condition for the irreversibility under PPT operations and construct a general class of such states. For this purpose, we consider an improved version of logarithmic negativity introduced in Ref.~\cite{Wang2016}, namely 
$$E_W(\rho)= \min  \log \|X^{T_B} \|_1, \mbox{s.t.} ~~X\ge\rho.$$ 
It was shown in Ref.~\cite{Wang2016} that  
$E_{D}(\rho)\le E_W(\rho)\le E_N(\rho)$,
and the second equality can be strict. It is straightforward to see that
if $E_W(\rho) < E_\eta(\rho)$, then $ E_{D}(\rho) < E_{C}(\rho). $

Indeed, we can obtain a more specific condition if we use logarithmic negativity $E_N$ instead of $E_W$. That is, for a bipartite state $\rho$, if there is a Hermitian matrix $Y$ such that $P_{AB}^{T_B} \pm Y\ge 0$ and $\|\rho^{T_B}\|_1 < \|Y^{T_B}\|_{\infty}^{-1}$, we have $E_{D}(\rho)<E_{C}(\rho)$.

We further show the irreversibility in asymptotic manipulations of entanglement under PPT operations  by a class of $3\otimes 3$ states in defined by 
$\rho^{(\a)}=( \proj{\psi_1} +\proj{\psi_2})/2$,
 where $\ket{\psi_1}=\sqrt{\a}\ket{01}-\sqrt{1-\a}\ket{10}$ and  $\ket{\psi_2}=\sqrt{\a}\ket{02}-\sqrt{1-\a}\ket{20}$  with $0.42\le\a \le 0.5$. Then the projection onto the range of $\rho^{(\a)}$ is $P_{AB}=\proj{\psi_1} +\proj{\psi_2}$.
One can easily calculate that
$$E_W(\rho^{(\a)})\le \log\|(\rho^{(\a)})^{T_B}\|_1=\log(1+\sqrt{2\a(1-\a)}).$$
We then construct a feasible solution to the dual SDP (\ref{E_eta}) of $E_\eta(\rho^{(\a)})$, i.e.,  
$Y=\a(\proj{01}+\proj{02})+(1-\a)(\proj{10}+\proj{20})+\sqrt{\a(1-\a)}(\proj{00}+\proj{11}+\proj{22}+\ketbra{11}{22}+\ketbra{22}{11})$. It can be checked that $-Y\le P_{AB}^{T_B}\le Y$ and $\| Y^{T_B}\|_{\infty}\le 1-\a$. Thus,  $E_{\eta}(\rho^{(\a)})\ge -\log(1-\a)$.

When $0.42\le \a \le 0.5$, it is easy to check that  $-\log(1-\a)>\log(1+\sqrt{2\a(1-\a)})$. Therefore,
$E_{D}(\rho^{(\a)})\le E_W(\rho^{(\a)}) < E_\eta(\rho^{(\a)})\le E_{C}(\rho^{(\a)})$.

\textbf{Discussions}
We prove that distillable entanglement can be strictly smaller than entanglement cost under PPT operations,  which implies the irreversibility of asymptotic entanglement manipulation under PPT operations. In particular, we prove that the PPT-distillable entanglement of any rank-two $3\ox3$ antisymmetric state is strictly smaller than its PPT-entanglement cost.  A byproduct is that there is a gap between the regularized Rains' bound and the asymptotic PPT-relative entropy of entanglement.  Consequently, there is no unique entanglement measure in general for the asymptotic entanglement manipulation under PPT operations, which indicates that entanglement theory under PPT operations differs from thermodynamics. We also obtain an SDP-computable lower bound for the entanglement cost under both LOCC and PPT operations.
Finally, we show an efficiently computable sufficient condition for the irreversibility of entanglement distillation of by LOCC (or PPT) operations. 

However, the lower bound $E_\eta$ for entanglement cost is in general not tight and could be sometimes smaller than distillable entanglement. To see this, consider the state $\sigma_a=P_a/3$ with $P_a$ the projection over the $3\ox 3$ antisymmetric subspace. We have $E_\eta(\sigma_a)=\log 3/2< E_{D}(\sigma_a)=E_{C}(\sigma_a)=\log5/3$~\cite{Audenaert2003}. How to further refine the lower bound $E_\eta$ remains an interesting problem.  

RD would like to thank Andreas Winter for inspirational discussions on the potential gap between regularized Rains' bound and asymptotic REE. The authors also thank Jonathan Oppenheim  and Martin Plenio for their helpful comments. This work was partly supported by the Australian Research Council under Grant Nos. DP120103776 and FT120100449.

%\newpage
%
%%
\newpage
\appendix
\widetext
\begin{center} \large{\textbf{Supplemental Material}} \end{center}

\section{The additivity of $E_\eta(\rho)$ under tensor product}
To see the additivity of  $E_\eta(\rho)$,  we reformulate it as $E_\eta(\rho)=-\log \eta(\rho)$, where 
\begin{equation}\label{dual  eta}\begin{split}
\eta(\rho)= \min &\ t \\
\text{ s.t. }&  -Y \le P^{T_B}\le Y,\\
& -t\1\le Y^{T_B}\le t\1,
\end{split}\end{equation}
where $P$ is the projection onto $\supp(\rho)$.

The dual SDP of $\eta(\rho)$ can be derived by  Lagrange multiplier method. It is given by
\begin{equation}\label{prime  eta}\begin{split}
\eta(\rho)= \max& \tr P(V-F)^{T_B},\\
 \text{ s.t. }&  V+F\le (W-X)^{T_B}, \\
& \tr(W+X)\le 1,\\
 & V,F,W,X\ge 0.
\end{split}\end{equation}
The optimal values of the primal and the dual SDPs above coincide by strong duality.  The details about strong duality theorem can be found in \cite{Watrous2011b}.

\begin{proposition}
For any two bipartite states $\rho_1$ and $\rho_2$, we have that  
 $$E_\eta(\rho_1\ox\rho_2)=E_\eta(\rho_1)+E_\eta(\rho_2).$$
\end{proposition}
\begin{proof}
On one hand, suppose that the optimal solution to SDP (\ref{dual eta}) of $\eta(\rho_1)$ and $\eta(\rho_2)$ are $ \{t_1,Y_1\}$ and $\{t_2,Y_2\}$, respectively. 
It is easy to see that  
\begin{align*}
 Y_1\ox Y_2+P_1^{T_B}\ox P_2^{T_{B'}}=\frac{1}{2}[(Y_1+P_1^{T_B})\ox (Y_2+P_2^{T_{B'}})+(Y_1-P_1^{T_B})\ox (Y_2-P_2^{T_{B'}})]\ge 0,\\
 Y_1\ox Y_2-P_1^{T_B}\ox P_2^{T_{B'}}=\frac{1}{2}[(Y_1+P_1^{T_B})\ox (Y_2-P_2^{T_{B'}})+(Y_1-P_1^{T_B})\ox (Y_2+P_2^{T_{B'}})]\ge 0.\\
\end{align*}
Then,  we have that $-Y_1\ox Y_2 \le P_1^{T_B}\ox P_2^{T_{B'}}\le Y_1\ox Y_2$. Moreover,
 $$\|Y_1^{T_B}\ox Y_2^{T_{B'}} \|_{\infty}\le \|Y_1^{T_B}\|_{\infty}\|Y_2^{T_{B'}}\|_{\infty}\le t_1t_2,$$
 which means that $-t_1t_2\1 \le Y_1^{T_B}\ox Y_2^{T_{B'}}\le t_1t_2\1$.
Therefore, $\{t_1t_2, Y_1\ox Y_2\}$ is a feasible solution to the SDP (\ref{dual eta}) of $\eta(\rho_1\ox \rho_2)$, which means that \begin{equation}\label{eta le}
\eta(\rho_1\ox\rho_2)\le t_1t_2=\eta(\rho_1)\eta(\rho_2).
\end{equation}

On the other hand, suppose that the optimal solutions to SDP (\ref{prime eta}) of $\eta(\rho_1)$ and $\eta(\rho_2)$ are $\{V_1, F_1,W_1,X_1\}$ and $\{V_2, F_2,W_2,X_2\}$, respectively. 
Assume that
\begin{align*}
V&=V_1\ox V_2+F_1\ox F_2, 
 F=V_1\ox F_2+F_1\ox V_2,\\
W&=W_1\ox W_2+X_1\ox X_2, 
 X=W_1\ox X_2+X_1\ox W_2.
\end{align*}
It is easy to see that
$$V+F=(V_1+F_1)\ox(V_2+F_2)
\le (W_1-X_1)^{T_B}\ox  (W_2-X_2)^{T_{B'}}
=(W-X)^{T_{BB'}}$$
and $\tr (W+X)=\tr(W_1+X_1)\ox(W_2+X_2)\le 1$. Thus, $\{V,F,W,X\}$ is a feasible solution to the SDP (\ref{prime eta}) of $\eta(\rho_1\ox \rho_2)$. This means that
\begin{equation}\label{eta ge}
 \eta(\rho_1\ox \rho_2)\ge \tr (P_1\ox P_2)(V-F)^{T_{BB'}}= \tr (P_1\ox P_2)((V_1-F_1)^{T_B}\ox(V_2-F_2)^{T_{B'}})
 = \eta(\rho_1)\eta(\rho_2).
\end{equation}

Hence, combining Eq. (\ref{eta le}) and Eq. (\ref{eta ge}), it is clear that $ \eta(\rho_1\ox \rho_2)= \eta(\rho_1)\eta(\rho_2)$, which means that $$E_\eta(\rho_1\ox\rho_2)=E_\eta(\rho_1)+E_\eta(\rho_2).$$
\end{proof}

\section{Proof of an inequality in Proposition \ref{E eta bound ER} using weak duality of SDP}
 In the following, we will utilize the weak duality of SDP to show an important inequality in Proposition \ref{E eta bound ER}, i.e.,
\begin{align}\label{ineq prop1}
\max_{\sigma_0\in \G} \tr P\sigma_0 &\le \min t \  \text{ s.t. } Y^{T_B}\le t\1, P^{T_B}\le Y.
\end{align}
To see this, we note that $\max_{\sigma_0\in \G} \tr P\sigma_0$ is the prime SDP and its dual can be derived by  Lagrange multiplier method. 
To be specific, we associate the operator  $G\ge 0$ to the constraint $\sigma_0^{T_B}\ge 0$ and a real multiplier $t$ to the constraint that $\tr \sigma_0=1$. The resulting Lagrangian is
\begin{align*}
\tr P \sigma_0
+ t(1-\tr \sigma_0)
+ \tr G\sigma_0^{T_B}
= \ t+\tr \sigma_0 (P+G^{T_B}-t\1).
\end{align*}
The dual SDP is to minimise $t$ subject to
\begin{equation}\label{cons appendix}
P+G^{T_B}-t\1\le 0, G\ge 0.
\end{equation}
Let $Y=P^{T_B}+G$, then the dual SDP is to minimise $t$ subject to
\begin{equation} 
Y^{T_B}\le t\1,  P^{T_B}\le Y.
\end{equation}
Therefore,  the prime and dual SDPs are as follows.
\begin{alignat}{2}
& \textbf{(Primal)} \quad && \max \left\{\tr P\sigma_0:  \sigma_0\ge 0, \sigma_0^{T_B}\ge 0, \tr\sigma_0=1 \right\}, \label{EW primal}\\
& \ \ \textbf{(Dual)} && \min \left\{ t :  Y^{T_B}\le t\1, P^{T_B}\le Y\right\}. \label{EW dual}
\end{alignat}
Finally, the inequality (\ref{ineq prop1}) follows from the weak duality theorem, 
which states that the value of the dual SDP attained at any dual feasible solution is at least the value of the primal SDP at any primal feasible solution. Interested readers can consult \cite{Vandenberghe1996,Watrous2011b} for more details. 

\section{Irreversibility for any rank-two $3\otimes 3$ antisymmetric state}
\begin{proposition}
For any state $\rho$ with spectral decomposition $$\rho=p\proj{u_1}+(1-p)\proj{u_2} \ (0<p<1),$$ where $$\ket {u_1}=(\ket{01}-\ket{10})/\sqrt2, \ket {u_2}=(\ket{ab}-\ket{ba})/\sqrt2,  $$
 it holds that $$ E_{D}(\rho)<1=E_{C}(\rho).$$
\end{proposition}
\begin{proof}
Suppose that $\ket a=a_0\ket0+a_1\ket1+a_2\ket2$ and $\ket b=b_0\ket0+b_1\ket1+b_2\ket2$. Noting that $\braket{u_1}{u_2}=0$, we have ${a_0}{b_1}-{a_1}{b_0}=0$. Thus, with simple calculation,  it is easy to see that
$$\ket {u_2}=[(a_0\ket0+a_1\ket1)\ox b_2\ket2 +a_2\ket2 \ox (b_0\ket0+b_1\ket1)-(b_0\ket0+b_1\ket1)\ox a_2\ket2-b_2\ket2 \ox (a_0\ket0+a_1\ket1)]/\sqrt2.$$
Then, one can simplify $\ket{u_2}$ to 
$$\ket {u_2}=[(\cos\theta\ket 0+\sin\theta\ket 1)\ox \ket2-\ket2\ox(\cos\theta\ket 0+\sin\theta\ket 1)]/\sqrt2 \ (0\le\theta\le \pi/2),$$
where $\theta$ is determined by $\ket a$ and $\ket b$. We assume that $P_{AB}=\proj{u_1}+\proj{u_2}$. 

It is can be calculated that $\|\rho^{T_B}\|_1<2$ for any $0<p<1$ and $0\le\theta\le \pi/2$, which means that $E_{D}(\rho) <1$.

Moreover, let us choose 
$$Y=P^{T_B}+\frac{1}{2}(\ket {00}+\ket{11}+\ket{22})(\bra {00}+\bra{11}+\bra{22}).$$ 
It is clear that $Y\ge P^{T_B}$ and it can be easily checked that $-Y\le P^{T_B}$. Thus, $Y$ is a feasible solution to the SDP (\ref{dual eta}) of $E_\eta(\rho)$, which means that 
$$E_\eta(\rho)\ge -\log \|Y^{T_B}\|_{\infty} =1.$$

Finally, it is clear that a standard $2\otimes 2$ maximally entangled state (one ebit) such as $1/\sqrt{2}(\ket{00}+\ket{11})$ is sufficiently to prepare an exact copy of $\rho$ via LOCC. Combining with the above lower bounds, we have that $1=E_{\eta}(\rho)\leq E_{R}^{\infty}(\rho)\leq E_{C}(\rho)\leq E_{C,LOCC}(\rho)\leq 1$.
\end{proof}
\end{document}